%% file: reach.tex
\documentclass[11pt]{article}

\usepackage{amsmath}
\usepackage{a4wide}
\usepackage{amssymb}
\usepackage{amsthm}

\usepackage[ruled,linesnumbered]{algorithm2e}

\usepackage{complexity}

\newclass{\USPACE}{USPACE}
\newclass{\TIUSP}{TIME{-}USPACE}
\newclass{\pUSP}{poly{-}USPACE}
\newclass{\qNC}{quasi{-}\NC}

\newtheorem{theorem}{Theorem}
\newtheorem{corollary}[theorem]{Corollary}
\newtheorem{lemma}[theorem]{Lemma}
\newtheorem{claim}[theorem]{Claim}
\theoremstyle{definition}

\newcommand{\dist}[4]{\mathrm{dist}_{#1}^{#2}(#3, #4)}
\newcommand{\mindist}[4]{\mathrm{min.dist}_{#1}^{#2}(#3, #4)}
\newcommand{\lrad}[4]{\mathrm{l}_{#1}^{#2}(#3, #4)}
\newcommand{\len}{\mathrm{len}}
\newcommand{\calp}[4]{\mathcal{P}_{#1}^{#2}(#3, #4)}
\newcommand{\calo}{\mathcal{O}}
\newcommand{\true}{\mathsf{true}}
\newcommand{\false}{\mathsf{false}}
\newcommand{\BW}{\mathsf{BAD.WEIGHT}}

\title{Trading Determinism for Time in Space Bounded Computations}
\author{Vivek Anand T Kallampally \footnote{Indian Institute of Technology Kanpur. {\tt email:vivekana@cse.iitk.ac.in}} \and Raghunath Tewari \footnote{Indian Institute of Technology Kanpur. {\tt email:rtewari@cse.iitk.ac.in}}}

\begin{document}
\maketitle
\begin{abstract}
Savitch showed in $1970$ that nondeterministic logspace (\NL) is contained in deterministic $\mathcal{O}(\log^2 n)$ space but his algorithm requires quasipolynomial time. The question whether we can have a deterministic algorithm for every problem in {\NL} that requires polylogarithmic space and simultaneously runs in polynomial time  was left open. 
		
In this paper we give a partial solution to this problem and show that for every language in {\NL} there exists an unambiguous nondeterministic algorithm that requires $\mathcal{O}(\log^2 n)$ space and simultaneously runs in polynomial time.
\end{abstract}
	
\input{intro}

\input{prelim}
\input{weight_assn}

\bibliographystyle{alpha}
\bibliography{references}
\end{document}

%% file: intro.tex
\section{Introduction}

Deciding reachability between a pair of vertices in a graph is an important computational problem from the perspective of space bounded computations. It is well known that reachability in directed graphs characterizes the complexity class nondeterministic logspace (\NL). For undirected graphs the problem was known to be hard for the class deterministic logspace (\L) and in a breakthrough result Reingold showed that is contained in {\L} as well \cite{Reingold08}. Several other restrictions of the reachability problem are known to characterize other variants of space bounded complexity classes \cite{Etessami97,Barrington89,BarringtonEtAl98}.

Unambiguous computations are a restriction of general nondeterministic computations where the Turing machine has at most one accepting computation path on every input. In the space bounded domain, unambiguous logspace (in short {\UL}) is the class of languages for which there is a nondeterministic logspace bounded machine that has a unique accepting path for every input in the language and zero accepting path otherwise. {\UL} was first formally defined and studied in \cite{BJLR91, AJ93}. In 2000 Reinhardt and Allender showed that the class {\NL} is contained in a non-uniform version of {\UL} \cite{RA00}. In a subsequent work it was shown that under the hardness assumption that deterministic linear space has functions that cannot be computed by circuits of size $2^{\epsilon n}$, it can be shown that $\NL=\UL$ \cite{ARZ99}. Although it is widely believed that {\NL} and {\UL} are the same unconditionally and in a uniform setting, the question still remains open.

Savitch's Theorem states that reachability in directed graphs is in $\DSPACE (\log^2 n)$, however the algorithm requires quasipolynomial time \cite{Sav70}. On the other hand standard graph traversal algorithms such as DFS and BFS can decide reachability in polynomial time (in fact linear time) but require linear space. Wigderson asked the question that can we solve reachability in $\calo (n^{1-\epsilon})$ space and polynomial time simultaneously, for some $\epsilon >0$ \cite{Wig92}. Barnes et. al. gave a partial answer to this question by giving a $\calo (n/2^{\sqrt{\log n}})$ space and polynomial time algorithm for the problem \cite{BBRS92}. Although this bound has been improved for several subclasses such as planar graphs \cite{INPVW13}, layered planar graphs \cite{CT15b}, minor-free and bounded genus graphs \cite{CPTVY14}, for general directed graphs (and hence for the class {\NL}) we still do not have a better deterministic space upper bound simultaneously with polynomial time.

\subsection{Main Result}

In this paper we show that directed graph reachability can be decided by an unambiguous $\calo (\log^2 n)$ space algorithm that simultaneously requires only polynomial time. Thus we get an improvement in the time required by Savitch's algorithm by sacrificing determinism. Formally, we show the following theorem.
\begin{theorem}
\label{thm:main}
$\NL \subseteq \pUSP(\log^2 n)$.
\end{theorem}
For the remainder of this paper all graphs that we consider are directed graphs unless stated otherwise. 

\subsection{Min-uniqueness of Graphs}

An important ingredient of our proof is the {\em min-uniqueness} property of graphs. A graph $G$ is said to be min-unique with respect to an edge weight function $W$ if the minimum weight path between every pair of vertices in $G$ is unique with respect to $W$. This turns out to be an important property and has been studied in earlier papers \cite{Wig94, GW96, RA00}. In fact, the fundamental component of Reinhardt and Allender's paper is a {\UL} algorithm for testing whether a graph is min-unique and then deciding reachability in min-unique graphs in {\UL} \cite{RA00}. They achieve this by proposing a {\em double inductive counting} technique which is a clever adaptation of the inductive counting technique of Immerman and Szelepcs\'{e}nyi \cite{Imm88,Sze88}. As a result of Reinhardt and Allender's algorithm, in order to show that reachability in a class of graphs can be decided in {\UL}, one only needs to design an efficient algorithm which takes as input a graph from this class and outputs an $ \calo (\log n) $ bit weight function with respect to which the graph is min-unique. This technique was successfully used to show a {\UL} upper bound on the reachability problem in several natural subclasses of general graphs such as planar graphs \cite{BTV07}, graphs with polynomially many paths from the start vertex to every other vertex \cite{PTV12}, bounded genus graphs \cite{DKTV11} and minor-free graphs \cite{AGGT16}. For the latter two classes of graphs reachability was shown to be in {\UL} earlier as well by giving reductions to planar graphs \cite{KV10, TW09}. Note that Reinhardt and Allender defines min-uniqueness for unweighted graphs where the minimum length path is unique, whereas we define it for weighted graphs where the minimum weight path is unique. However it can easily be seen that both these notions are equivalent.

\subsection{Overview of the Proof}

We prove Theorem \ref{thm:main} in two parts. We first show how to construct an $\calo(\log^2 n)$ bit weight function $W$ with respect to which the input graph $G$ becomes min-unique. Our construction of the weight function $W$ uses an iterative process to assign weights to the edges of $G$. We start by considering a subgraph of $G$ having a fixed radius and construct an $\calo (\log n)$ bit weight function with respect to which this subgraph becomes min-unique. For this we first observe that there are polynomially many paths in such a subgraph and then use the prime based hashing scheme of Fredman, Koml\'{o}s and Szemer\'{e}di \cite{FKS84} to give distinct weights to all such paths. Thereafter, in each successive round of the algorithm, we construct a new weight function with respect to which a subgraph of double the radius of the previous round becomes min-unique and the new weight function has an additional $\calo (\log n)$ bits. Hence in $\calo (\log n)$ many rounds we get a weight function which has $\calo (\log^2 n)$ bits and with respect to which $G$ is min-unique. We show that this can be done by an unambiguous, polynomial time algorithm using $\calo (\log^2 n)$ space. This technique is similar to the isolating weight construction in \cite{FGT16}, but their construction is in {\qNC}.

We then show that given a graph $G$ and an $\calo (\log^2 n)$ bit weight function with respect to which $G$ is min-unique, reachability in $G$ can be decided by an unambiguous, polynomial time algorithm using $\calo (\log^2 n)$ space. Note that a straightforward application of Reinhardt and Allender's algorithm will not give the desired bound. This is because ``unfolding'' a graph with $\calo (\log^2 n)$ bit weights will result in a quasipolynomially large graph. As a result we will not achieve a polynomial time bound. We tackle this problem by first observing that although there are $2^{\calo (\log ^2 n)}$ many different weight values, the weight of a shortest path can only use polynomial number of distinct such values.  Using this observation we give a modified version of Reinhardt and Allender's algorithm that iterates over the ``good'' weight values and ignores the rest. This allows us to give a polynomial time bound.

The rest of the paper is organized as follows. In Section \ref{sec:prelim} we define the various notations and terminologies used in this paper. We also state prior results that we use in this paper. In Section \ref{sec:result} we give the proof of Theorem \ref{thm:main}.

%% file: prelim.tex
\section{Preliminaries}
\label{sec:prelim}
	
For a positive integer $n$, let $[n] = \{1,2, \ldots , n\}$. Let $G = (V, E)$ be a directed graph on $n$ vertices and let $E=\{e_{1}, e_{2}, \ldots, e_{m}\}$ be the set of edges in $G$. Let $s$ and $t$ be two fixed vertices in $G$. We wish to decide whether there exists a path from $s$ to $t$ in $G$. The {\em length} of a path $P$ is the number of edges in $P$ and is denoted as $\len (P)$. The {\em center} of a path $P$ is a vertex $x$ in $P$ such that the length of the path from either end point of $P$ to $x$ is at most $\lceil \len(P)/2 \rceil$ and $ x $ is no farther from the tail of $ P $ than from the head of $ P $. 
	
A {\em weight function} $w : E \rightarrow \mathbb{N}$ is a function which assigns a positive integer to every edge in $G$. The weight function $w$ is said to be {\em polynomially bounded} if there exists a constant $k$ such that $w(e) \leq \calo (n^k)$ for every edge $e$ in $G$. We use $ G_{w} $ to denote the weighted graph $G$ with respect to a weight function $ w $. For a graph $ G_{w} $, the {\em weight of a path} $ P $ denoted by $ w(P) $ is defined as the sum of weights of the edges in the path. A {\em shortest path} from $ u $ to $ v $ in $ G_{w} $ is a path from $ u $ to $ v $ with minimum weight. Let $ \calp{w}{i}{u}{v} $ denote the set of shortest paths from $u$ to $v$ of length at most $i$ in $G_w$. Thus in particular, the set of shortest paths from $ u $ to $ v $ in $ G_{w} $, $ \calp{w}{}{u}{v} = \calp{w}{n}{u}{v} $. 
We define the {\em distance} function with respect to a weight function and a nonnegative integer $i$ as 
\[ \dist{w}{i}{u}{v} = \left\{ \begin{array}{ll}
         w(P) & \textrm{ for }P \in \calp{w}{i}{u}{v}\\
                 \infty & \textrm{ if } \calp{w}{i}{u}{v} = \emptyset \end{array} \right. \] 

Correspondingly we define the function $l$  which represents the minimum length of such paths as 
\[ \lrad{w}{i}{u}{v} = \left\{ \begin{array}{ll}
         \min_{P \in \calp{w}{i}{u}{v}} \{ \len (P)\} & \textrm{ if } \calp{w}{i}{u}{v} \ne \emptyset\\
                 \infty & \textrm{ otherwise} \end{array} \right. \] 
		
A graph $ G_{w} $ is said to be {\em min-unique} for paths of length at most $ i $, if for any pair of vertices $u$ and $v$, the shortest path from $ u $ to $ v $ with length at most $i$, is unique. $ G_{w} $ is said to be min-unique if $ G_{w} $ is min unique for paths of arbitrary length. Define weight function 
\[w_{0}(e_{i}) :=  2^{i -1}, \textrm{ where } i \in [m].\] 
It is straightforward to see that for any graph $ G $, $w_{0}$ is an $ n $ bit weight function and $ G_{w_{0}} $ is min-unique. Wherever it is clear from the context that there is only one weight function $w$, we will drop the subscript $ w $ in our notations. 
	
For a graph $ G_{w} $, vertex $ u $ in $G$, length $ i $ and weight value $ k $, we define the quantities $ c_{k}^{i}(u)$ and $D_{k}^{i}(u)$ as the number of vertices at a distance at most $k$ from $u$, using paths of length at most $i$ and the sum of the distances to all such vertices respectively. Formally, 
	\begin{align*}
		c_{k}^{i}(u)  &= | \{v \mid \dist{w}{i}{u}{v} \leq k\} | \\
		D_{k}^{i}(u) &= \sum_{v \mid \dist{w}{i}{u}{v} \leq k}^{} \dist{w}{i}{u}{v}.
	\end{align*}
	
	An {\em unambiguous Turing machine} is a nondeterministic Turing machine that has at most one accepting computation path on every input \cite{VAL76}. We shall consider unambiguous computations in the context of space bounded computations. $\USPACE(s(n))$ denotes the class of languages decided by an unambiguous machine using $ \calo(s(n)) $ space. In particular, $\UL = \USPACE(\log n)$. $\TIUSP(t(n), s(n))$ denotes the class of languages decided by an unambiguous machine using $ \calo(s(n)) $ space and $ \calo(t(n)) $ time simultaneously. In particular, when $t(n)$ is a polynomial, we define 
\[\pUSP(s(n))=\bigcup_{k \geq 0} \TIUSP(n^k,s(n)).\]

For graphs having polynomially many paths, we use the well known hashing technique due to Fredman, Koml\'{o}s and Szemer\'{e}di \cite{FKS84} to compute a weight function that assigns distinct weights to all such paths. We state the result below in a form that will be useful for our purpose.

	\begin{theorem} \label{thm:hashing}
		\cite{FKS84, PTV12} For every constant $ c $ there is a constant $ c' $ so that for every set $ S $ of
		$ n $ bit integers with $ |S| \leq n^{c} $ there is a $ c' \log n$ bit prime number $ p $ so that for all $ x \neq y \in S, \ x \not\equiv y \bmod{p} $.
	\end{theorem}

Henceforth we will refer to Theorem \ref{thm:hashing} as the FKS hashing lemma.

%% file: weight_assn.tex
\section{Min-unique Weight Assignment}
\label{sec:result}

	Reinhardt and Allender \cite{RA00} showed that for every $n$ there is a sequence of $n^2$ $\calo (\log n)$ bit weight functions such that every graph $G$ on $n$ vertices is min-unique with respect to at least one of them. For each weight function they construct an unweighted graph (say $G_w$) by replacing every edge with a path of length equal to the weight of that edge. Since the weights are $\calo (\log n)$ bit values therefore $G_w$ is polynomially large in $n$. Next they show that using the double inductive counting technique one can check unambiguously using a logspace algorithm if $G_w$ is min-unique, and if so then check if there is a path from $s$ to $t$ as well.  They iterate over all weight functions until they obtain one with respect to which $G_w$ is min-unique and use the corresponding graph $G_w$ to check reachability. Since we use an $\calo (\log^2 n)$ bit weight function with respect to which the input graph is min-unique, we cannot construct an unweighted graph by replacing every edge with a directed path of length equal to the corresponding edge weight.

In Section \ref{sec:wt} we give an algorithm that computes an $\calo (\log^2 n)$ bit, min-unique weight function and decides reachability in directed graphs. In Section \ref{sec:check} we check if a graph is min-unique. Although we use $\omega (\log n)$ bit weight functions, our algorithm still runs in polynomial time. In Section \ref{sec:dist} we show how to compute the $\dist{w}{i}{u}{v}$ function unambiguously.

\subsection{Construction of the weight function}
\label{sec:wt}

Theorem \ref{thm:wt} shows how to construct the desired weight function.

\begin{theorem}
\label{thm:wt}
There is a nondeterministic algorithm that takes as input a directed graph $G$ and outputs along a unique computation path, an $\calo (\log^{2} n) $ bit weight function $ W $ such that $G_W$ is min-unique, while all other computation paths halt and reject. For any two vertices $ s $ and $ t $ the algorithm also checks whether there is a path from $ s $ to $ t $ in G. The algorithm uses $\calo(\log^2 n)$ space and runs in polynomial time.
\end{theorem}
Since directed graph reachability is complete for {\NL}, Theorem \ref{thm:main} follows from Theorem \ref{thm:wt}.
  	\begin{algorithm}[h]
  		\SetAlgoNoLine
  		\DontPrintSemicolon
  		\SetKwFor{For}{for}{do}{endfor}
  		\SetKwFor{ForEach}{for each}{do}{endfor}
  		\SetKwIF{If}{ElseIf}{Else}{if}{then}{else if}{else}{endif}
  		\KwIn{($G, s, t)$ }
  		\KwOut{weight function $W:=W_q$, $\true$ if there is a path from $s$ to $t$ and $\false$ otherwise}
  		\Begin{
  			$ q := \log n $; $ W_{0} := 0 $ \;
  			\For{$j \leftarrow 1$ \KwTo $ q $}{
  				$ i := 2^{j} $; $ p := 2 $ \;
  				\Repeat{$ (G, W_{j}, i) $ is min-unique}
  				{
  					\tcc{ By the FKS hashing lemma $p$ is bounded by a polynomial in $n$, say $n^{c'}$. We define $B := n^{c'+2}$. }
  					$W_{j} := B \cdot W_{j-1} + (w_{0} \bmod  p) $ \;
  					Check whether $ (G, W_{j}, i) $ is min-unique using Algorithm \ref{algo:minunique} \;
  					$ p := $ next prime \;
  				}
   			}
  			\lIf{$ \dist{W_{q}}{n}{s}{t} \leq  B^{q}$}{\Return ($W_q, \true$)}
				\lElse{\Return ($W_q$, $\false$)}
  		}
  		\caption{Computes a min-unique weight function and checks for an $s-t$ path in $G$}
  		\label{algo:final}
  	\end{algorithm}
\begin{proof}[Proof of Theorem \ref{thm:wt}]
To prove Theorem \ref{thm:wt} we design an algorithm that outputs the desired weight function. The formal description of the construction is given in Algorithm \ref{algo:final}. The algorithm works in an iterative manner for $\log n$ number of rounds. Initially we consider all paths in $G$ of length at most $l$ where $l=2^1$. The number of such paths is bounded by $n^l$ and therefore by the FKS hashing lemma there exist a $ c' \log n $ bit prime $ p_{1} $ such that with respect to the weight function $ W_{1} := w_{0} \bmod p_{1} $, $ G_{w_{1}} $ is min-unique for paths of length at most $ l $. To find the right prime $ p_{1} $ we iterate over all  $c' \log n$ bit primes and use Lemma \ref{lem:minunique} to check whether $ G_{w_{1}} $ is min-unique for paths of length at most $ l $.  

We prove this by induction on the number of rounds, say $j$. Assume that $G_{W_{j-1}}$ is min-unique for paths of length at most $2^{j-1}$. In the $j$-th round, the algorithm considers all paths of length at most $2^j$. By applying Lemma \ref{lem:induction} we get a weight function $W_j$  from $W_{j-1}$ which uses $\calo (j \cdot \log n)$ bits and $G_{W_j}$ is min-unique for paths of length at most $2^j$. Hence in $\log n$ many rounds we get a weight function $W:= W_{\log n}$ such that $G_W$ is min-unique. Note that the inner {\bf repeat-until} loop runs for at most $n^{c'}$ iterations due to the FKS hashing lemma.

Let $p_j$ be the prime used in the $j$-th round of Algorithm \ref{algo:final}. Define $p' := \max \{p_j \mid j \in [\log n]\}$. By the FKS hashing lemma $p'$ is bounded by a polynomial in $n$, say $n^{c'}$. We set $B := n^{c'+2}$. This implies that for any weight function of the form $w = w_0 \bmod p_j$ and any path $P$ in $G$, $w(P) < B$. Observe that with respect to the final weight function $W$, for any path $P$ in $G$, $W(P) < B^q$.

Once we compute an $ \calo (\log^{2} n) $ bit weight function $ W $ such that $ G_{W} $ is min-unique, there exist a path from $ s $ to $ t $ if and only if $ \dist{W}{n}{s}{t} \leq  B^{q} $. This can be checked using Algorithm \ref{algo:dist} in $ \calo(\log^{2} n) $ space and polynomial time. Also Algorithm \ref{algo:dist} is a nondeterministic algorithm which returns $ \true $ or $ \false $ along a unique computation path while all other computation paths halt and reject.

In each round the size of $ W_{j} $ increases by $ \calo(\log n) $ bits and after $ \log n $ rounds $ W_{\log n} $ is an $\calo(\log^{2} n) $ bit weight function. By Lemma \ref{lem:minunique} checking whether a graph is min-unique with respect to an $\calo(\log ^{2} n) $ bit weight function requires $ \calo(\log ^{2} n) $ space. Thus the total space complexity of Algorithm \ref{algo:final} is $ \calo(\log ^{2} n) $.

The FKS hashing lemma guarantees that in each round only a polynomial number of primes need to be tested to find a weight function which is min-unique for paths of length at most $ 2^{j} $. By Lemma \ref{lem:minunique} checking whether a graph is min-unique for paths of length at most $ 2^{j} $ can be done in polynomial time. Thus each round runs in polynomial time. There are only $ \log n $ many round and hence Algorithm \ref{algo:final} runs in polynomial time.

By Lemma \ref{lem:minunique}, Algorithm \ref{algo:minunique} is a nondeterministic algorithm which outputs its answer along a unique computation path, while all other computation paths halt and reject. All other steps in Algorithm \ref{algo:final} are deterministic. This shows the unambiguity requirement of the theorem.
\end{proof}

	\begin{lemma}
	\label{lem:induction}
	There is a nondeterministic algorithm $\mathcal{A}$, that takes as inputs $ (G, w) $ where $G$ is a graph on $n$ vertices and $w$ is a $k$ bit weight function such that $G_{w}$ is min-unique for paths of length at most $l$. $\mathcal{A}$ outputs a $(k+ \calo (\log n))$ bit weight function $w'$ such that $G_{w'}$ is min-unique for paths of length at most $2l$, along a unique computation path while all other computation paths halt and reject. $\mathcal{A}$ uses $\calo(k+ \calo (\log n))$ space and runs in polynomial time.
	\end{lemma}
		The encoding of the output weight function $ w' $ is the concatenation of the $ k $ bit representation of the input weight function $ w $ and an $ \calo(\log n) $ bit prime number $ p $. The output weight function $ w' $ is calculated as  $ w' := B \cdot w + w_{0} \bmod p $, where $ B $ is the number defined in Algorithm $ \ref{algo:final} $. Multiplication using $ B $ is used just to left shift $ w $ and make room for the new function $w_{0} \bmod p$. 	
	
	Lemma \ref{lem:induction} proves the correctness of each iteration of the outer {\bf for} loop of Algorithm \ref{algo:final}. Before proving the lemma, we will show that if $G_w$ is min-unique for paths of length at most $l$, then the number of minimum weight paths with respect to $w$ of length at most $2l$ is bounded by a polynomial independent of $l$. Hence it allows us to use the FKS hashing lemma to isolate such paths.
	
	\begin{lemma} \label{lem:bound}
		Let $G$ be a graph with $n$ vertices and $w$ be a weight function such the graph $ G_{w} $ is min-unique for paths of length at most $l$. Then for any pair of vertices $ u $ and $ v $, $ \left | \calp{w}{2l}{u}{v} \right | $ is at most $n$.
	\end{lemma}

	\begin{proof}
		Let $P$ be a shortest path from $u$ to $v$ in $ G_{w} $ with length at most $2l$ with center vertex $x$. That is $ P \in \calp{w}{2l}{u}{v} $. Let $P_{1}$ and $P_{2}$ be the subpaths from $u$ to $x$ and $x$ to $v$. Since $x$ is the center of $P$, $P_{1}$ has length at most $l$. Note that $P_{1}$ is the unique shortest path of length at most $l$ from $u$ to $x$ in $ G_{w} $. This is because if there exists another path of length at most $l$ with a smaller weight than $ P_{1} $ from $u$ to $x$ then replacing $ P_{1} $ with this path in $P$ will result in a path of length at most $2l$ from $u$ to $v$ with a lower weight than $P$. But this cannot happen since $P$ is a shortest path from $u$ to $v$.
				
		\begin{claim}
			There is only one shortest path of length at most $2l$ from $u$ to $v$ with $x$ as its center.
		\end{claim}
		\begin{proof}
			Assume there is another shortest path $P'$ of length at most $ 2l $ from $u$ to $v$ with $x$ as its center. Let $P_1'$ be the subpath of $P'$ from $u$ to $x$. Since $x$ is the center of $P'$, $ P'_{1} $ is of length at most $l$. Similar to $ P_{1} $, $P'_{1}$ is a shortest path of length at most $l$ from $u$ to $x$. This means there are two shortest paths of length at most $l$ from $u$ to $x$. This is a contradiction since $G$ is min-unique for paths of length at most $l$. 
		\end{proof}
		Therefore each vertex can be the center of at most one path of length at most $ 2l $ from $ u $ to $ v $. Thus the total number of shortest paths of length at most $ 2l $ from $ u $ to $ v $ in $ G_{w} $ is at most $n$. Hence $ \left | \calp{w}{2l}{u}{v} \right | \leq n$. This completes the proof of Lemma \ref{lem:bound}.
	\end{proof}
		When we sum over all possible pairs of $ u $ and $ v $, the total number of shortest paths of length at most $2l$ in $G_{w}$ is at most $n^{3} $.

	\begin{proof}[Proof of Lemma \ref{lem:induction}]
		$G_{w}$ is min-unique for paths of length at most $l$. Therefore by Lemma \ref{lem:bound} the number of shortest paths between all pairs of vertices with at most $2l$ edges in $G$ is at most $ n^{3} $. Let $\mathcal{S}$ be the set of these $n^{3}$ shortest paths. With respect to the weight function $w_{0}$ (see Section \ref{sec:prelim}) each element of $\mathcal{S}$ gets a distinct weight. So by using the FKS hashing lemma we get a constant $c'$ and a $ c' \log n $ bit prime number $ p $ such that with respect to the weight function $\widehat{w}$ such that $ \widehat{w} := w_{0} \bmod p $, each element of $\mathcal{S}$ gets a distinct weight. Moreover, in $G$ between any pair of vertices the shortest path in $\mathcal{S}$ is unique.
		
		Let $B$ be the number as defined in Algorithm \ref{algo:final}. Now consider the weight function $w' := B \cdot w + \widehat{w}$. Since $w$ is a $k$ bit weight function and $\widehat{w}$ is an $\calo (\log n)$ bit weight function therefore $w'$ is a $(k+ \calo (\log n))$ bit weight function. Clearly $ w $ has higher precedence than $ \widehat{w} $ in $ w' $. So for any two paths $ P_{1} $ and $ P_{2} $ in $ G $ , we have if $ w'(P_{1}) < w'(P_{2})$ then either $ w(P_{1}) < w(P_{2})$ or both the predicates $w(P_{1}) = w(P_{2})$ and $\widehat{w}(P_{1}) < \widehat{w}(P_{2})$ are true. Additionally if $ w'(P_{1}) = w'(P_{2})$ then $w(P_{1}) = w(P_{2})$ and $\widehat{w}(P_{1}) = \widehat{w}(P_{2})$.

		All the unique shortest paths of length at most $ 2l $ in $ G_{w} $, will be unique shortest paths of length at most $ 2l $ in $ G_{w'} $ also. If there are multiple shortest paths of length at most $ 2l $ from $ u $ to $ v $ in $ G_{w} $, $\widehat{w}$ gives a unique weight to each of these paths. So $ G_{w'} $ is min-unique for paths of length at most $ 2l $.
		
		We can check whether a graph $ G_{w'} $ is min-unique for paths of length at most $ 2l $ using Lemma \ref{lem:minunique}. Since $ p $ is an $ c' \log n $ bit prime number, we can iterate over all the $ c' \log n $ bit primes and find $p$.	
	\end{proof}

\subsection{Checking for min-uniqueness}
\label{sec:check}

The next lemma shows how to check whether $G_w$ is min-unique for paths of length at most $l$ in an unambiguous manner. 
		\begin{algorithm}[h]
		\SetAlgoNoLine
		\DontPrintSemicolon
		\SetKwFor{For}{for}{do}{endfor}
		\SetKwFor{ForEach}{for each}{do}{endfor}
		\SetKwIF{If}{ElseIf}{Else}{if}{then}{else if}{else}{endif}
		\KwIn{($G, w, i $)}
		\KwOut{$\true$ if $G_w$ is not min-unique for paths of length at most $i$ and $\false$ otherwise}
		\Begin{
			$\BW := \false$ \;
			\tcc{$\BW$ is set to $\true$ whenever the weight function does not make the graph min-unique. Otherwise it remains $\false$. It is a boolean variable shared between Algorithms \ref{algo:induct} and \ref{algo:minunique} } 
			\ForEach{vertex $v$} {
				$c_{0}^{i}(v) := 1;$ $D_{0}^{i}(v) := 0;$ $k' := 0$ \;
				\Repeat { $\BW = \true$ } {
					$k := k';$ $c_{k}^{i}(v) := c_{k'}^{i}(v);$ $D_{k}^{i}(v) := D_{k'}^{i}(v)$ \;
					Find next $k'$ from ($G, w, v, i, k, c_{k}^{i}(v), D_{k}^{i}(v)$) using Algorithm \ref{algo:nextk} \;
					\lIf{$ k' = \infty $} {
						break
					}
					Compute ($c_{k'}^{i}(v), D_{k'}^{i}(v)$) from ($G, w, v, i, k, c_{k}^{i}(v), D_{k}^{i}(v), k'$) using Algorithm \ref{algo:induct} \;
				}
				\lIf {$\BW = \true$} {
					break
				}
			}
			\Return $\BW$ \;
		}		
		\caption{Check whether $G$ is min-unique for paths of length at most $i$} 
		\label{algo:minunique}
	\end{algorithm}

	\begin{lemma} \label{lem:minunique}
		There is a nondeterministic algorithm that takes as input a directed graph $G$, a $ k $ bit weight function $ w $ and a length $ i $ and outputs along a unique computation path whether or not the graph $ G_{w} $ is min-unique for paths of length at most $ i $, while all other computation paths halt and reject. The algorithm uses $\calo(k + \log n)$ space and runs in polynomial time. 
	\end{lemma}
	
	For every vertex $v$ in the $G_w$ we check whether there are two minimum weight paths of length at most $i$ to some other vertex in $G$. Algorithm \ref{algo:minunique} gives a formal description of this process. The algorithm iterates over all shortest path weight values that can be achieved by some path of length at most $i$. 
		
In the $k$-th stage of the algorithm it considers a ball of radius $k$ consisting of vertices which have a shortest path of weight at most $k$ from $v$ and length at most $i$. $c_{k}^{i}(v)$ denotes the number of vertices in this ball and $D_{k}^{i}(v)$ denotes the sum of the weights of the shortest paths to all such vertices. Initially $k = 0$, $c_{0}^{i}(v) =1$ (consisting of only the vertex $v$) and $D_{0}^{i}(v)=0$. 

A direct implementation of the double inductive counting technique of Reinhardt and Allender \cite{RA00} does not work since this would imply that we cycle over all possible weight values, which we cannot afford. We bypass this hurdle by considering only the relevant weight values. We compute the immediate next shortest path weight value $k'$, and use $k'$ as the weight value for the next stage of the algorithm. This computation is implemented in Algorithm \ref{algo:nextk}). Lemma \ref{lem:nextk} proves the correctness of this process. Note that the number of shortest path weight values from a fixed vertex is bounded by the number of vertices in the graph. This ensure that the number of iterations of the inner {\bf repeat-until} loop of Algorithm \ref{algo:minunique} is bounded by $n$. 

	\begin{algorithm}[h]
		\SetAlgoNoLine
		\DontPrintSemicolon
		\SetKwFor{For}{for}{do}{endfor}
		\SetKwFor{ForEach}{for each}{do}{endfor}
		\SetKwIF{If}{ElseIf}{Else}{if}{then}{else if}{else}{endif}
		\KwIn{($G, w, u, i, k, c_{k}^{i}(u), D_{k}^{i}(u)$)}
		\KwOut{ $ k' := \min \{ \dist{w}{i}{u}{v} \mid \dist{w}{i}{u}{v} > k, \ v \in V \} $ }
		\Begin{
			$k' := \infty$ \;
			\ForEach{vertex $v$} {
				\If {$\neg (\dist{w}{i}{u}{v} \leq k)$} {
					$\mindist{w}{i}{u}{v} := \infty$ \;
					\ForEach {$x$ such that $(x, v)$ is an edge} {
						\If {$\dist{w}{i}{u}{x} \leq k$ and $ \lrad{w}{i}{u}{x} + 1 \leq i $} {
							\If {$\mindist{w}{i}{u}{v} > \dist{w}{i}{u}{x} + w(x,v)$} {
								$\mindist{w}{i}{u}{v} := \dist{w}{i}{u}{x} + w(x,v)$ \;
							}
						}
					}
					\lIf {$k' > \mindist{w}{i}{u}{v}$} {
						$k' := \mindist{w}{i}{u}{v}$
					}
				}
			}
			\Return $k'$ \; 
		}
		\caption{Find the next smallest weight value $k' > k$ among all paths of length at most $ i $ from $ u $}
		\label{algo:nextk}
	\end{algorithm}

\begin{lemma}
\label{lem:nextk}
Given $(G, w, u, i, k, c_{k}^{i}(u), D_{k}^{i}(u))$, Algorithm \ref{algo:nextk} correctly computes the value $\min \{ \dist{w}{i}{u}{v} \mid \dist{w}{i}{u}{v} > k, \ v \in V \} $.
\end{lemma}
To see the correctness of Lemma \ref{lem:nextk} observe that for every vertex $v$ such that $\dist{w}{i}{u}{v} > k$, the algorithm cycles through all vertices $x$ such that there is an edge from $x$ to $v$ and the length of the path from $u$ to $x$ is at most $i-1$. It computes the minimum weight of such a path and store it in the variable $\mindist{w}{i}{u}{v}$. It then computes the minimum value of $\mindist{w}{i}{u}{v}$ over all possible vertices $v$ and outputs it as $k'$, as required.

After we get the appropriate weight value $k'$, we then compute the values of $c_{k'}^{i}(v)$ and $D_{k'}^{i}(v)$ by using a technique similar to Reinhardt and Allender (implemented in Algorithm \ref{algo:induct}). Additionally we also maintain a shared flag value $\BW$ between Algorithms  \ref{algo:minunique} and \ref{algo:induct}, which is set to $\true$ if $G_w$ is not min-unique for paths of length at most $i$, else it is $\false$.

	\begin{algorithm}[h]
		\SetAlgoNoLine
		\DontPrintSemicolon
		\SetKwFor{For}{for}{do}{endfor}
		\SetKwFor{ForEach}{for each}{do}{endfor}
		\SetKwIF{If}{ElseIf}{Else}{if}{then}{else if}{else}{endif}
		\KwIn{($G, w, u, i, k, c_{k}^{i}(u), D_{k}^{i}(u), k'$)}
		\KwOut{($c_{k'}^{i}(u), D_{k'}^{i}(u)$) and also flag $\BW$ }
		\Begin{
			$c_{k'}^{i}(u) := c_{k}^{i}(u)$; $D_{k'}^{i}(u) := D_{k}^{i}(u)$ \;
			\ForEach{vertex $v$} {
				\If {$\neg (\dist{w}{i}{u}{v} \leq k)$} {
					\ForEach {$x$ such that $(x, v)$ is an edge} {
						\If {$\dist{w}{i}{u}{x} \leq k$ and $\dist{w}{i}{u}{x} + w(x,v) = k'$ and $ \lrad{w}{i}{u}{x} + 1 \leq i $} {
							$c_{k'}^{i}(u):= c_{k'}^{i}(u) + 1$; $D_{k'}^{i}(u) := D_{k'}^{i}(u)+k'$ \;
							\ForEach {$x' \neq x$ such that ($x', v$) is an edge} {
								\If { $\dist{w}{i}{u}{x'} \leq k$ and $\dist{w}{i}{u}{x'} + w(x',v) = k'$ and $ \lrad{w}{i}{u}{x'} + 1 \leq i $} {
									$\BW := \true$ \;
								}
							}
						}
					}
				}
			}
			\Return ($c_{k'}^{i}(u), D_{k'}^{i}(u)$)
		}
		\caption{Compute $c_{k'}^{i}(u)$ and $D_{k'}^{i}(u)$ and check whether $ G_{w} $ is min-unique for paths with length at most $ i $ and weight at most $ k' $ from $ u $} 
		\label{algo:induct}
	\end{algorithm}

\subsection{Computing the $\dist{w}{i}{u}{v}$ function}
\label{sec:dist}

In Algorithms \ref{algo:nextk} and \ref{algo:induct}, an important step is to check whether $\dist{w}{i}{u}{v} \leq k$ and if so, get the values of $\dist{w}{i}{u}{v}$ and $\lrad{w}{i}{u}{v}$. These values are obtained from Algorithm \ref{algo:dist}. Algorithm \ref{algo:dist} describes a nondeterministic procedure that takes as input a weighted graph $G_w$, which is min-unique for paths of length at most $i$ and weight at most $ k $ from a source vertex $ u $ and the values $c_k^i (u)$ and $D_k^i (u)$. For any vertex $ v $, if $\dist{w}{i}{u}{v} \leq k$ then it outputs $\true$ and the values of $\dist{w}{i}{u}{v}$ and $\lrad{w}{i}{u}{v}$ along a unique computation path. Otherwise it outputs $\false$ along a unique computation path with $\infty$ as the values of $\dist{w}{i}{u}{v}$ and $\lrad{w}{i}{u}{v}$. All other computation paths halt and reject. As a result we can compute the predicate $\neg (\dist{w}{i}{u}{v} \leq k)$ along a unique path as well.

	\begin{algorithm}[h]
		\SetAlgoNoLine
		\DontPrintSemicolon
		\SetKwFor{For}{for}{do}{endfor}
		\SetKwFor{ForEach}{for each}{do}{endfor}
		\SetKwIF{If}{ElseIf}{Else}{if}{then}{else if}{else}{endif}
		\KwIn{($G, w, u, i, k, c_{k}^{i}(u), D_{k}^{i}(u), v$)}
		\KwOut{ ($\true$ or $\false$), $\dist{w}{i}{u}{v}$, $ \lrad{w}{i}{u}{v} $ }
		\Begin{
			$count := 0$; $sum := 0$; $path.to.v :=  \false$ \;
			$\dist{w}{i}{u}{v} := \infty$; $ \lrad{w}{i}{u}{v}  := \infty $\;
			\ForEach{$x \in V$}{
				Guess non deterministically if $\dist{w}{i}{u}{x} \leq k$ in $ G_{w} $\; \label{alg_line:guess}
				\If {the guess is $\dist{w}{i}{u}{x} \leq k$ } {
					Guess a path of weight $d \leq k$ and length $l \leq i$ from $u$ to $x$ \; \label{alg_line:guesspath}
					(If this fails then halt and reject) \;
					$count := count + 1$; $sum := sum + d$ \;
					\If { $x = v$} {
						$path.to.v := \true$ \;
						$ \dist{w}{i}{u}{v} := d $ \;
						$ \lrad{w}{i}{u}{v} := l $ \;
					}
				}
				
			}\label{alg_line:for}
			\eIf{$count = c_{k}^{i}(u)$ and $sum = D_{k}^{i}(u)$} {\label{alg_line:if}
				\Return ($path.to.v$, $\dist{w}{i}{u}{v}$, $ \lrad{w}{i}{u}{v} $) \; 
			}{
				halt and reject \;
			}
		}
		\caption{An unambiguous routine to determine if $\dist{w}{i}{u}{v} \leq k$ and find $\dist{w}{i}{u}{v}$ and $ \lrad{w}{i}{u}{v} $}
		\label{algo:dist}
	\end{algorithm}

Note that Algorithm \ref{algo:dist} is the only algorithm where we use non-determinism. The algorithm is similar to the unambiguous subroutine of Reinhardt and Allender \cite{RA00} with the only difference being that here we consider weight of a path instead of length of a path. The algorithm assumes that the subgraph induced by all the paths of length at most $ i $ and weight at most $ k $ from $ u $ is min-unique. 

In Line \ref{alg_line:guess} of Algorithm \ref{algo:dist}, for each vertex $ x $ the routine non-deterministically guesses whether $ \dist{w}{i}{u}{x} \leq k $ and if the guess is `$ \true $', it then guesses a path of length at most $ k $ from $ u $ to $ x $. If the algorithm incorrectly guesses for some vertex $ x $ that $ \dist{w}{i}{u}{x} > k $, then the variable $ count $ will never reach $ c_{k}^{i}(u) $ and the routine will reject. If it guesses incorrectly that $ \dist{w}{i}{u}{x} \leq k $ it will fail to guess a correct path in Line \ref{alg_line:guesspath} and again reject that computation. Thus the only computation paths that exit the {\bf for} loop in Line \ref{alg_line:for} and satisfy the first condition of the {\bf if} statement in Line \ref{alg_line:if}, are the ones that correctly guess exactly the set $ \{x \mid \dist{w}{i}{u}{x} \leq k\} $. If the algorithm ever guesses incorrectly the weight $ d $ of the shortest path to $ x $, then if $ \dist{w}{i}{u}{x} > d $ no path of weight $ d $ will be found, and if $ \dist{w}{i}{u}{x} < d $ then the variable $ sum $ will be incremented by a value greater than $ \dist{w}{i}{u}{x} $. In the latter case, at the end of the algorithm, $ sum $ will be greater than $ D_{k}^{i}(u) $, and the routine will reject.

Since $ G_{w} $ is min-unique for paths of length at most $ i $ and weight at most $ k $ from $ u $, only for exactly one computation path $ sum $  and $ count $ will match with $ c_{k}^{i}(u) $ and $ D_{k}^{i}(u) $. So except one computation path which made all the guesses correct, all other paths halt and reject. If $ \dist{w}{i}{u}{v} \leq k $ then even though the algorithm uses non-deterministic choices, it outputs `$ \true $' along a single computation path while all other paths halt and reject. Also if $ \dist{w}{i}{u}{v}  > k $, the algorithm outputs `$ \false $' along a single computation path while all other paths halt and reject. The space complexity of the algorithm is bounded by the size of the weight function $w$.

As a corollary of Theorem \ref{thm:main} we get the following result.
\begin{corollary}
For $s(n) \geq \log n$, $\NSPACE(s(n)) \subseteq \TIUSP(2^{\calo (s(n))}, s^{2}(n))$.
\end{corollary}